\documentclass[12pt]{article}
\usepackage{amsfonts}
\usepackage{amssymb}
\usepackage{amscd}
\usepackage{amsmath}
\usepackage{authblk}

\usepackage{amsthm}

\usepackage{bbm}

\usepackage[usenames]{color}
\usepackage{appendix}

\newcommand{\BEQ}{\begin{equation}}
\newcommand{\EEQ}{\end{equation}}
\def\bea{\begin{eqnarray}}
\def\eea{\end{eqnarray}}
\def\nn{\nonumber}

\newtheorem{Th}{Theorem}
\newtheorem{lemma}{Lemma}
\newtheorem{Prop}{Proposition}

\newtheorem{remark}{Remark}

\newcommand{\beq}[1]{\begin{equation}\label{#1}}
\newcommand{\eq}{\end{equation}}

\def\bea{\begin{eqnarray}}
\def\eea{\end{eqnarray}}

\def\bes{\begin{equation*} \begin{split}}
\def\ees{\end{split} \end{equation*}}

\def\R{{\mathbb{ R}}}

\newcommand{\g}{\ensuremath{\mathfrak{g}}}
\newcommand{\gl}{\ensuremath{\mathfrak{gl}}}
\newcommand{\sln}{\ensuremath{\mathfrak{sl}_n}}
\newcommand{\bn}{\ensuremath{\mathfrak{b}_n}}

\usepackage{graphicx}

\begin{document}

\title{ Full symmetric Toda system: \\ QR-solution for complete DLNT-family}
\author[1,2,3]{Yury B. Chernyakov\footnote{chernyakov@itep.ru}}
\author[1,2,4]{Georgy I. Sharygin\footnote{sharygin@itep.ru}}
\author[1,4,5]{Dmitry V. Talalaev\footnote{dtalalaev@yandex.ru}}
\affil[1]{\small NRC "Kurchatov institute", Kurchatov square, 1, 123182, Moscow, Russia.
}
\affil[2]{\small  Moscow Institute of Physics and Technology, 141700, Dolgoprudny, Russia.
}
\affil[3]{\small Institute for Information Transmission Problem, RAS, Bolshoy Karetny per. 19, build.1, 127051, Moscow, Russia.
}
\affil[4]{\small  Lomonosov Moscow State University,
119991, Moscow, Russia.
}
\affil[5]{\small Demidov Yaroslavl State University, 150003, Sovetskaya Str. 14
}

\date{\today}
%
%
%

\maketitle
\begin{abstract}
The paper is devoted to the algebraic and geometric aspects of the full symmetric Toda system. We construct a solution to the complete Deift-Li-Nanda-Tomei flows system using the QR decomposition method. For this purpose we introduce specialized invariant tensor operations on the Lax operator of the model. These operations have a direct interpretation in terms of the representation theory of Lie algebras. We expect that this approach can be effective in studying the geometry of flag varieties.
\end{abstract}
\tableofcontents
\section{Introduction}
Ever since its discovery in mid 1960-ies (although reportedly before that it showed up as a curious non-chaotic behaviour in the Ulam-Fermi numerical experiment in 1950-ies), the Toda system (open or periodic Toda chain and its generalisations) has almost always remained in the spotlight of attention of the specialists in Mathematical Physics. Seemingly simple way to formulate it, on one hand, and the fact that it shows up in one or another manner in virtually any study of integrable systems, from Matrix models to Painlev\'e equations, makes it an indispensable tool and a central object of interest in many works. The fact that it can be easily generalised to embrace different phase spaces, its relation with algebra and geometry of the Lie groups and Lie algebras, give further incentives for the research of its properties.

In this paper we are dealing with one of the popular generalizations of the usual Toda chain, that is with the full symmetric Toda system; putting it simply one can say that the symmetric Toda system differs from the usual Toda chain by the form of the Lax matrix (see section \ref{sect:todadef1}): whereas in the usual Toda chain, this matrix is symmetric and has three non-zero diagonals, the full symmetric system works with generic symmetric matrices, see equation \eqref{eq:laxtoda}.

This system is known to be Hamiltonian and integrable in an appropriately generalised sense of Liouville integrability criterion: the phase space of this system is not symplectic, so one needs to speak either about some sort of the noncommutative integrability (see \cite{Neh}) or about the integrability on symplectic leaves. The latter fact is far from being trivial, since the dimensions of symplectic leaves are rather large, in fact they are proportional to $n^2$; the constructions of additional integrals are scarce and rather complicated (see \cite{DLNT},\,\cite{CS}). Not all these families are commutative; in effect one can show that the symmetric Toda system is super-integrable. In our paper we try to shed some light on the genesis and the structure of the best known commutative family of additional Toda integrals, the \textit{chopping integrals} of Deift, Nando, Li and Tomei (see \cite{DLNT}). One knows that they can be obtained by means of the AKS reduction method from a family of $B_-$-invariant rational functions on $\sln^*$. On the other hand, the fact that these functions actually commute on $\sln^*$ is not easy to show: one either does it by direct computations (as in \cite{DLNT}), or one can derive it  from the properties of a Gaudin system, see \cite{T1}. In our paper we derive this property from the application of a differential operator to the invariant functions. We think, that this mechanism of obtaining commutative families, similar to the ``argument shift method'' is an interesting and a rather easy way of reasoning, which deserves a special attention.

Another subject, which we address in this paper is the way, one can obtain the solutions to the additional flows, i.e. to the Hamilton equations associated with the additional (chopping) integrals. It has been known in the literature, that the symmetric Toda flow can be solved by the method of ``QR-decomposition'': i.e. if $L_0$ is the initial position of the flow, and we decompose the exponent of $L_0$ in the product of an orthogonal and a lower-triangular matrix
\[
\exp(tL_0)=Q(t)R(t)
\]
then the curve $Q(t)L_0Q(t)^{-1}$ is the integral curve of the Toda system that goes through $L_0$. In our paper we show that similar constructions are applicable to the chopping integrals. To this end we use a special property of the Lax matrix, associated with these integrals, see sections \ref{sect:qrdec1} and \ref{sect:chopqr1}. We think, the fact that the DLNT integral flows can be solved by this method can be helpful if one wants to understand the geometry of this flows and the corresponding coordinates on the phase spaces.

\subsection{Composition and main results of the paper}
The body of paper is divided into four major parts: first, in section \ref{sect:prelim}, we speak about the known results concerned with the symmetric Toda system, beginning with the description of the center of the Kitillov-Kostant Poisson structure on $\sln^*$. Then we give a brief definition of the Toda system in terms of the Lax matrix, and speak a little about the known systems of first integrals: the chopping procedure. We also outline the major steps of the AKS reduction method, that gives commutative families of functions on Lie subalgebras from some invariant functions on the big algebra.

Then in section \ref{sect:keylem12}, we prove the commutativity of the chopping integrals by a method, based on the application of a differential operators $D_k$, see the lemmas \ref{lem:key} and \ref{lem:key2}. We think, that the role of these operators in the general theory is rather large and we hope to find further applications of these constructions in the future. We also give a description of another family of additional integrals (this family is noncommutative in general).

Section \ref{sect:laxqr} is the main part of the text. In this section we show that the QR decomposition method can be extended to the additional integrals, such as the chopping integrals. To this end we show that the Hamilton fields of these integrals are generated by the $M$-operators, that verify certain important additional condition, see the equation \eqref{eq:nablaeki}. This is done by finding a suitable matrix representation of the $M$-operator: first we do it in a particular case of the chopping integrals $H_{k,n-1}(L)$ (see sections \ref{sect:qrdec1} and \ref{sect:hamilttk}), and then in the general case, where the operators $D_k$ come up very handy in the reasoning, see section \ref{sect:chopqr1}.

The last section is dedicated to a list of open questions, concerned with the Toda system. In the end we add the appendix, where some of the computations, referred to in the paper are collected.

\paragraph{Acknowledgements} The work on sections 1,2 and 3 was carried out within the framework of a development programme for the Regional Scientific and Educational Mathematical Center of the Yaroslavl State University with financial support from the Ministry of Science and Higher Education of the Russian Federation (Agreement on provision of subsidy from the federal budget No. 075-02-2022-886) and was partially supported by the RFBR grant 20-01-00157. The work on sections 4,5,and 6 was partially supported by the RSCF grant 22-11-00272.
\section{Preliminaries}\label{sect:prelim}
In this section we describe the known results and constructions, related with the full symmetric Toda system and its integrals.
\subsection{Poisson algebra $S(\mathfrak{sl}_n)$}
Let $\mathfrak g$ be a Lie algebra, $\mathfrak g^*$ its dual space. The algebra of polynomial functions on the affine space $\mathfrak g^*$ is naturally identified with $S(\mathfrak g)$. Recall that the Lie bracket on a Lie algebra $\mathfrak g$ induces a Poisson structure in $S(g)$; it is well known that the Poisson center of the algebra $S(\mathfrak g)$ consists of the $G$-invariant elements in $S(\mathfrak g)$, that is the polynomial functions preserved by the (co)adjoint action on $\mathfrak g^*$ of the simply connected group $G$, generated by $\mathfrak g$.

For instance when $\mathfrak g=\mathfrak{sl}_n$, the Poisson structure is defined by the formula
\begin{equation}\label{eq:poikks}
\{e_{ij},e_{kl}\}=\delta_{jk}e_{il}-\delta_{li}e_{kj}.
\end{equation}
Here $e_{ij}\in S(\mathfrak{sl}_n),\,\sum_ie_{ii}=0$ are the generators of the Lie algebra $\mathfrak{sl}_n$. It is convenient to consider the matrix (the Lax operator) made up of $e_{ij}$
\bea
L=\sum_{ij} E_{ij}\otimes e_{ij},\nn
\eea
then $L\in Mat_n(S(\mathfrak{sl}_n)$ and one can write down the Poisson structure in an $R$-matrix form.

Another advantage of using the matrix $L$ is that it gives a simple way to describe the Poisson center of $S(\mathfrak{sl}_n)$: as we mentioned above $\mathcal Z(S(\mathfrak{sl}_n))$ is equal to the space of all $SL_n$-invariant functions in $S(\mathfrak{sl}_n)$. Then the generators of this subalgebra can be easily identified with the traces of the powers of the matrix $L$:
\bea
\mathcal Z(S(\mathfrak{sl}_n))=\langle I_m(L),\,m=2,\dots,n\rangle,\nn
\eea
where $\langle\cdot\rangle$ denotes the subalgebra, generated by $\cdot$ and
\bea
I_m=\frac{1}{m}Tr(L^m).\nn
\eea
In fact, the elements $I_m$ are independent, and the Poisson center of $S(\mathfrak{sl}_n)$ is just free polynomial algebra with generators $I_k$. Another way to get the generators of the Poisson center of $S(\mathfrak{sl}_n)$ is to consider the coefficients of the characteristic polynomial of $L$:
\begin{equation}\label{eq:generate1}
\det(L-\lambda\mathbbm 1)=\sum_{i=0}^n (-1)^{n-i}\lambda^{n-i}E_i(L),
\end{equation}
then $E_1=Tr(L)=0$ in \sln\ and
\[
\mathcal Z(S(\mathfrak{sl}_n))=\langle E_i(L),\,i=2,\dots,n\rangle.
\]
The elements $E_i(L)$ are again free generators of the center of $S(\mathfrak{sl}_n)$.

\subsection{Full symmetric Toda system and Borel subalgebras}\label{sect:todadef1}
Consider the Cartan decomposition of a real semisimple Lie algebra $\mathfrak g$:
\begin{equation}\label{eq:cartandec}
\g=\mathfrak{k}\oplus\mathfrak p,
\end{equation}
where $\mathfrak{k}$ is a maximal compact subalgebra in \g\ and $\mathfrak p$ is the subspace, orthogonal complement of $\mathfrak k$ with respect to the Killing form. The decomposition in \eqref{eq:cartandec} is determined with the help of Cartan involution $\theta:\g\to\g,\,\theta^2=\mathbbm 1$, so that $\theta_{|_{\mathfrak k}}=\mathbbm 1,\,\theta_{|_{\mathfrak p}}=-\mathbbm 1$. For instance for $\g=\sln$ (all the algebras in this section are assumed to be real) we have
\[
\sln=\mathfrak{so}_n\oplus Symm_n(\mathbb R)
\]
and $\theta(a)=-a^t$. Here $Symm_n(\mathbb R)$ denotes the space of symmetric $n\times n$ matrices, and $a^t$ is the transpose matrix. Then the Killing form on \g\ induces an isomorphism:
\begin{equation}\label{eq:identybor}
\mathfrak k\cong\mathfrak k^*,\, \mathfrak p=(\mathfrak k)^\perp\cong\mathfrak b^*
\end{equation}
Here $\mathfrak{b}$ denotes the (lower) Borel subalgebra in \g\ and $\mathfrak b^*$ is its dual space. For instance in the case $\g=\sln,\,\mathfrak b=\bn$, the algebra of lower triangular matrices in $\mathfrak{sl}_n$.

The identification \eqref{eq:identybor} induces the Poisson structure on $\mathfrak p$, it is this structure, which is used in the definition of the full symmetric Toda system. In particular, if $\g=\sln$, we have the Poisson structure pulled from the space $\bn^*$.  From now on we shall concentrate on this particular case, although the most part of our results are easy to generalise to arbitrary Cartan decompositions.

So let us consider the Toda system on \sln. By definition this is the flow, induced by the Poisson structure pulled from \bn\ for the Hamiltonian $H=\frac12Tr(L^2)$ (here we use the notation from the previous section). One can show that the corresponding Hamilton equation has the following Lax form
\begin{equation}\label{eq:laxtoda}
\dot L=[M(L),L],
\end{equation}
where $M:Symm_n\to\mathfrak{so}_n$ is the natural projection: if the symmetric matrix $L$ is equal to the sum of a diagonal matrix and the upper and lower triangular matrices $L_+,\,L_-$ (where $L_-=L_+^T$, of course) then
\begin{equation}\label{eq:moper}
M(L)=L_+-L_-.
\end{equation}
\begin{remark}\label{rem:moper1}\rm
The projection $M$ is in effect equal to the restriction onto $Symm_n$ of the natural projection $M:\sln\to\mathfrak{so}_n$, induced by the direct sum decomposition:
\[
\sln=\mathfrak{so}_n\oplus\bn.
\]
\end{remark}

\subsection{Chopping procedure in the Full Symmetric Toda system}\label{sect:chop1}
The identifications that we made above mean that the integrability of Toda system now will follow if we find a sufficient number of integrals in involution commuting with the Hamiltonian $H=\frac12Tr(L^2)$, regarded as a function on $\bn^*$. One such family is easy to obtain: it is not difficult to show that the center of the Poisson algebra $S(\sln)$ gives such a family, when restricted to $Symm_n$. In effect it is almost evident that all these functions will commute with the Hamilton function $H$: this is just the corollary of the Lax equation \eqref{eq:laxtoda}. The fact that these functions commute with each other follows from the next formula for the Hamilton field of a smooth function $f:Symm_n\to\R$ with respect to the Poisson structure, induced from $\mathfrak b_n$:
\begin{equation}\label{eq:hamfieldsym1}
X_f=[M(\nabla f),L]
\end{equation}
where $M$ is the projector to $\mathfrak{so}_n$ that we have used earlier, see \eqref{eq:moper}. Here for any smooth function on $\mathfrak{sl}_n$, we use $\nabla f$ to denote the matrix of partial derivatives of $f$: let $x_{ij},\,i,j=1,\dots,n$ denote the matrix elements of $L$; also let $\partial_{ij}=\frac{\partial}{\partial x_{ij}}$, then
\begin{equation}\label{eq:nabladef}
\nabla f=\begin{pmatrix}
\partial_{11}f &\dots &\partial_{1n}f\\
\vdots &\ddots &\vdots\\
\partial_{n1}f &\dots &\partial_{nn}f
\end{pmatrix}.
\end{equation}
Thus, the generators $E_m(L),\,m=1,\dots,n$ (coefficients of the characteristic polynomial of $L$) or $I_m(L)=\frac{1}{m}Tr(L^m),\,m=1,\dots,n$ of the center, of $S(\sln)$ give a family of the first integrals of the Toda system. However they are not sufficient in the great majority of cases, since the dimension of the space $Symm_n$ is $n(n-1)/2-1$ much greater than the number of these functions.

This means that one needs additional integrals to prove the integrability of the symmetric Toda system. The main construction, used in the literature to find such additional integrals is the so-called \textit{chopping construction}, or \textit{chopping procedure} suggested in the paper \cite{DLNT}. Let us recall it briefly.

Consider the matrix of the Lax operator $L$ of the Full Symmetric Toda system; it is a real symmetric matrix of order $n$. For future references we will fix the notation here:
\beq{Lax} L = \left(
\begin{array}{c c c c c c}
 a_{11} & a_{12} & ... & a_{1n}\\
 a_{12} & a_{22} & ... & a_{2n}\\
 ... & ... & ... & ...\\
 a_{1n} & a_{2n} & ... & a_{nn}\\
\end{array}
\right).
\eq
Observe that in accordance with traditions we use letters $a_{ij}$ and not $x_{ij}$ to denote the matrix elements of $L$ as symmetric matrix i.e.
\[
a_{ij}=(x_{ij})_{|_{Symm_n}}.
\]
One can now define the following set of characteristic polynomials:
\beq{Pkn}
\begin{array}{c}
\chi_{k}(L,\lambda) = \det\left((L - \lambda \mathbbm 1)^{(k)}\right),\\
\ \\
\chi_{k}(L,\lambda) = \sum_{m=0}^{n-2k} E_{k,m}(L)\lambda^{n-2k-m}, \,\,\, 0 \leq k \leq \left[\frac{n}{2}\right],
\end{array}
\eq
where $(L - \lambda \mathbbm 1)^{(k)}$ is the matrix of order $n-k$, which we obtain by deleting (``chopping out'') the $k$ upper rows and $k$ right columns of matrix $(L - \lambda \mathbbm 1)$, and $[ \ \ \ ]$ means the floor integer part; below we will usually abbreviate $\chi_{k}(L,\lambda)$ to just $\chi_{k}(\lambda)$. This procedure was introduced in paper \cite{DLNT}, where it's called the \textit{chopping procedure}. It was shown in the cited paper that the functions
\beq{Integrals}
I_{k,m}(L) = \frac{E_{k,m}(L)}{E_{k,0}(L)}, \ \ \ 0 \leq m \leq \left[\frac{1}{2}(n-1)\right],  \ \ \ 1 \leq k \leq n-2m
\eq
define $\left[\frac{1}{4} n^{2}\right]$ integrals in involution for the Full Symmetric Toda system on the generic orbit of order $2\left[\frac{1}{4} n^{2}\right]$. These integrals are functionally independent, we will call them the \textit{DLNT-integrals} or the \textit{chopping integrals}.

The notation for $E_{k,m}(L)$ is consistent with the previously introduced one: recall that we denoted by denote by $E_{m}(L)$ the coefficients of characteristic polynomials of the ``unchopped'' matrix. Then $E_m(L) =E_{0,m}$ (in fact $E_{0,0}(L)=1$), so the previously introduced integrals, that appear from the center of $S(\sln)$ are part of the family. Of course we can replace $E_m(L)$ with a more standard set of functions $I_{m}(L)=\frac{1}{m} Tr L^{m},\,m=1,\dots,n$.

The original proof of commutativity of the family $I_{k,m}(L)$ as functions on $Symm_n$ with respect to the Poisson structure, induced from \bn, given in \cite{DLNT}, was by a brute-force calculation. Below (see lemma \ref{lem:key}) we shall give an independent prove of the commutativity of $I_{k,m}(L)$; our approach is rather invariant and with its help we will be able to prove that the flow induced by the ``chopping integrals'' $I_{k,m}(L)$ admits the solution by QR-decomposition method.

\subsection{AKS lemma}
In order to obtain Poisson commutative families of rational functions in $Symm_n$, or (equivalently) in $Q(\mathfrak{b}_n)$, the algebra of \textit{rational} functions on $\mathfrak{b}_n^*$ we will use a universal construction, known as the Adler-Kostant-Symes (AKS) method (see \cite{Ad}, \cite{K1} and \cite{S1}). Let us recall the basic steps of this method. To this end consider a Lie algebra $\mathfrak g$, decomposed into the direct sum
\begin{equation}\label{eq:decompaks}
\mathfrak g=\mathfrak g_0\oplus\mathfrak g_1.
\end{equation}
Let $p:\mathfrak{g}\to\mathfrak{g}_0$ be the natural projection, induced by the decomposition \eqref{eq:decompaks}. We will denote by $S(p)$ the natural homomorphism of the symmetric algebras
\bea
S(p):S(\mathfrak{g})\rightarrow S(\mathfrak{g}_0),\nn
\eea
induced by $p$. In effect, $S(p)$ is just the projection with respect to the natural decomposition
\begin{equation}\label{eq:aks2}
S(\mathfrak{g})=S(\mathfrak{g}_0)\oplus\mathfrak{g}_1S(\mathfrak{g}).
\end{equation}
Observe that both terms in this decomposition are in fact Poisson subalgebras in $S(\mathfrak g)$. If $i:\mathfrak{g}_0\to\mathfrak{g}$ is the inclusion of the subalgebra, then $S(i)$ will denote the induced map on symmetric algebras. Clearly, $S(i)$ unlike $S(p)$ commutes with the Poisson brackets on $S(\mathfrak{g})$ and $S(\mathfrak{g}_0)$; in fact the first term in the decomposition \eqref{eq:aks2} is the image of the map $S(i)$; below we will usually identify $S(\mathfrak{g}_0)$ with the image of the inclusion map $S(i)$.

For an element $f\in S(\mathfrak g)$ let us denote its decomposition with respect to \eqref{eq:aks2} as
\bea
f=f_0 +f_1\nn
\eea
where $f_0=S(p)(f)$. The following lemma gives a way to construct a commutative subalgebras in $S(\mathfrak g_0)$.
\begin{lemma}\label{lem:aks}
Let $f,g\in S(\mathfrak{g})^{G_0}=S(\mathfrak{g})^{\mathfrak g_0}$ (here $G_0$ is the simply connected Lie group associated with $\mathfrak g_0$) and let $\{f,g\}=0$ in $S(\mathfrak g)$, then
\[
\{f_0,g_0\}=0=\{f_1,g_1\}.
\]
\end{lemma}
\begin{proof}
Indeed, for $f,g\in S(\mathfrak{g})^{G_0}$ we compute
\bea
\{f_1, g_1\} = \{f - f_0, g - g_0\} = \{f, g\} - \{f_0, g\} - \{f, g_0\} + \{f_0, g_0\} = \{f_0, g_0\}.
\label{eq:aks}
\eea
Here $\{f,g\}=0$ by assumption and $\{f_0,g\}=0=\{f,g_0\}$ due to the condition that both $f$ and $g$ are $\mathfrak{g}_0$-invariant. But the opposite sides of \eqref{eq:aks} lie in different non-intersecting subalgebras of $S(\mathfrak{g})$ (in $S(\mathfrak{g}_0)$ on the right and in $\mathfrak{g}_1S(\mathfrak{g})$ on the left), hence both are zeroes.
\end{proof}
Clearly, this reasoning does not depend on the kind of functions that we use; below we shall apply this schema to the rational functions on $\mathfrak g$, i.e. to the ratios of two polynomials, rather than to the polynomial ones.

\section{First integrals of the Toda system}\label{sect:keylem12}
\subsection{Full symmetric Toda and AKS method}
The full symmetric Toda system naturally fits into the context of AKS scheme for the decomposition
\begin{equation}\label{eq:triangledec}
\mathfrak{sl}_n= \mathfrak{b}_n\oplus\mathfrak{so}_n.
\end{equation}
It turns out that in this case the map
\[
S(p)(S(\mathfrak{sl}_n)^{\mathfrak{b}})\rightarrow S(\mathfrak b_n)
\]
from the AKS scheme induces a maximal commutative subalgebra of integrals of the Toda system. By dimension counts, to prove this it is enough to show that the chopping integrals belong to its image, which we will now do.

So now we are going to construct a commutative family of $\mathfrak b_n$-invariant rational functions on $\mathfrak{sl}_n^*$. To this end consider the family of subalgebras $\mathfrak p_n^k\subseteq \mathfrak{sl}_n$, where
\[
\mathfrak p_n^k=\left\{X=(x_{ij})\in\mathfrak{sl}_n\mid x_{ij}=0,\ \mbox{if}\ i=1,2,\dots,k,\ \mbox{or}\ j=n-k+1,n-k+2,\dots,n\right\}.
\]
In other words, $\mathfrak p_n^k$ consists of the matrices $X\in\mathfrak{sl}_n$ of the form
\begin{equation}\label{eq:matrixx}
X=\begin{pmatrix}
0 & 0 & \dots & 0 & 0 &\dots & 0\\
\vdots & \vdots & \ddots & \vdots & \vdots &\ddots & \vdots\\
0 & 0 & \dots & 0 & 0 & \dots & 0\\
x_{k+1,1} & x_{k+1,2} &\dots & x_{k+1,n-k} & 0 &\dots & 0\\
x_{k+2,1} & x_{k+1,2} &\dots & x_{k+1,n-k} & 0 &\dots & 0\\
\vdots & \vdots & \ddots & \vdots & \vdots &\ddots & \vdots\\
x_{n1} & x_{n2} &\dots & x_{n,n-k} & 0 &\dots & 0
\end{pmatrix}
\end{equation}
Clearly, $\mathfrak p_n^k\subset \mathfrak p_n^q$ if $k\ge q$. Also remark that the algebra $\mathfrak p_n^k$ is commutative as long as $k\ge\left[\frac{n}{2}\right]$. In effect, sequence of subalgebras $\{\mathfrak p_n^k\},\,k=1,\dots,n$ is nested:
\[
\mathfrak{sl}_n=\mathfrak p_n^0\supset\mathfrak p_n^1\supset\mathfrak p_n^2\supset\dots\supset \mathfrak p_n^n,
\]
so the subalgebra $\mathcal A$ generated by the Poisson centres of all $\mathfrak p_n^k,\,k=1,\dots,n$ is a Poisson commutative subalgebra in $S(\sln)$. Here is a construction, that relates this subalgebra with the AKS method and the solutions of full symmetric Toda system on \sln. Let $P^k:\mathfrak{sl}_n\to\mathfrak p_n^k\subseteq \mathfrak{sl}_n$ be the natural projection to $\mathfrak p_n^k$; in somewhat informal terms $P^k$ consists of ``deleting'' the first $k$ rows and the last $k$ columns in the matrix. 
\begin{lemma}\label{lem:key}
Let $E_2,\dots,E_n$ be the second set of generators of the Poisson center in $S(\sln)$, see \eqref{eq:generate1}. Let $D_k$ be the differential operator
\[
\frac{\partial}{\partial x_{1n}}\frac{\partial}{\partial x_{2,n-1}}\dots\frac{\partial}{\partial x_{k,n-k+1}}:S(\sln)\to S(\sln).
\]
Then for all $k$ and $i$ the elements $E_{k,i}=D_k(E_i)$ commute with each other. In effect, they are in the subalgebra $\mathcal A$.
\end{lemma}
The commutativity of the chopping functions $E_{k,i}$ (see the remark \ref{rem:chopdef} below for their original definition) was originally proved by Deift, Nando, Li and Tomei by a straightforward computation; the role of parabolic subalgebras of \sln\ was further observed by different authors, first of all for the full Kostant-Toda system (see \cite{ErcolFlaSin}). Alternatively, this property can be derived as a special case of the Gaudin system integrals, see \cite{T1}. Here we give a proof of the commutativity both for the sake of completeness of our text and also because later we shall use some of the properties of this construction. Besides this, we believe that the approach based on the use of differential operators is not only simple and elegant, but can be of much help if one wants to further generalise this construction.
\begin{proof}

Observe that $D_k=\frac{\partial}{\partial x_{k,n-k+1}}\circ D_{k-1}$. Also it's clear that $D_k(E_i)\in S(\mathfrak p_n^k)\subseteq S(\sln)$ for all $i$ and $k$: indeed $E_k$ depends linearly on every variable $x_{pq}$ and every monomial in $E_k$ contains only one element from each row and column of this matrix, see equation \eqref{eq:generate1}. Hence by induction it is sufficient to show that $\frac{\partial}{\partial x_{k,n-k+1}}(f)$ Poisson-commutes with all elements in $S(\mathfrak p_n^k)$ for $f\in\mathcal Z(S(\mathfrak p_n^{k-1}))$. 

So we take any $g\in S(\mathfrak p_n^{k}),\,f\in\mathcal Z(S(\mathfrak p_n^{k-1}))$ and compute:
\begin{equation}\label{eq:computecommute}
0=\frac{\partial}{\partial x_{k,n-k+1}}\{f,g\}=\left\{\frac{\partial f}{\partial x_{k,n-k+1}},g\right\}+\left\{f,\frac{\partial g}{\partial x_{k,n-k+1}}\right\}+\{f,g\}_{k,n-k+1}.
\end{equation}
Here $\{,\}_{k,n-k+1}$ denotes the Poisson bracket with respect to the bivector $\frac{\partial}{\partial x_{k,n-k+1}}\pi$, where $\pi$ is the Poisson bivector, that determines the structure \eqref{eq:poikks}; a direct computation shows
\[
\{f,g\}_{i,j}=\sum_{p=1}^n\left\{\frac{\partial f}{\partial x_{i,p}}\frac{\partial g}{\partial x_{p,j}}-\frac{\partial f}{\partial x_{p,j}}\frac{\partial g}{\partial x_{i,p}}\right\}
\]
for all $i,j=1,\dots,n$. Now taking $i=k,\,j=n-k+1$ we see that $\{f,g\}_{k,n-k+1}=0$ since $\frac{\partial g}{\partial x_{p,n-k+1}}=0=\frac{\partial g}{\partial x_{k,p}}$ for all $p$, since $g\in S(\mathfrak p_n^{k})$ so it doesn't depend on these variables, hence the second and the third terms in the equation \eqref{eq:computecommute} vanish and we have: $\left\{\frac{\partial f}{\partial x_{k,n-k+1}},g\right\}=0$.
\end{proof}
\begin{remark}\label{rem:chopdef}\rm
It turns out that the elements $D_k(E_i)$ do up to a sign coincide with the numerators and denominators of the ``chopping integrals'' due to Deift and others, \cite{DLNT}, see section \ref{sect:chop1}: the matrix $(L-\lambda\mathbbm 1)^{(k)}$ is in fact equal to the nonzero part of the projection $P^k(L-\lambda\mathbbm 1)$. Then
\begin{equation}\label{eq:altchop1}
\chi_k(\lambda)=\det\left((L-\lambda\mathbbm 1)^{(k)}\right)=(-1)^{(n-1)k}D_k(\det(L-\lambda\mathbbm 1))=(-1)^{(n-1)k}\sum_{i=0}^n (-1)^{n-i}\lambda^{n-i}D_k(E_i).
\end{equation}
The fact that the degree in $\lambda$ of the polynomial drops by $2k$ and not by $k$ under the action of $D_k$, as one may expect, follows from the condition that when one differentiate by $x_{\ell,n-\ell+1}$ only the monomials that contain this variable will survive.
\end{remark}
\begin{lemma}\label{lem:key2}
The elements $E_{k,i}$ are semi-invariants of the natural $\mathfrak b_n$-action, i.e. there exists a character $c_k:\mathfrak b_n\to \mathbb C$ for which $E_{k,i}^b=c_k(b)E_{k,i}$ for all $b\in\mathfrak b_n$; moreover the character $c_k$ does not depend on $i$.
\end{lemma}
\begin{proof}
To show this, we observe that since $E_{m},\,m=1,\dots,n$ is in the Poisson center of $S(\mathfrak{sl}_n)$ so it is an $SL_n$-invariant polynomial. Hence it is in particular $\mathfrak b_n$-invariant.We have
\[
E_{k,i}^b=(D_k(E_{i}))^b=D_k^b(E_{i}^b)=D_k^b(E_{i}),
\]
where $D_k=\frac{\partial}{\partial x_{1n}}\frac{\partial}{\partial x_{2,n-1}}\dots\frac{\partial}{\partial x_{k-1,n-k+1}}$ is the operator that we defined above. Now, the action of $\mathfrak b_n$ on the partial differential operators with constant coefficients on $\mathfrak{sl}_n$, can be identified with the minus the action of $\mathfrak b^T_n$ (the conjugate Borel subalgebra of $\mathfrak b_n$ in $\mathfrak{sl}_n$, i.e. the algebra of upper triangular matrices) on $\mathfrak{sl}_n$ so that $\frac{\partial}{\partial x_{pq}}$ corresponds to the generator $e_{pq}$; this follows from the identification of the partial derivation operation $\frac{\partial}{\partial x_{pq}}$ with the derivation of the free commutative algebra $S(\sln)$, induced by the pairing of $\sln$ with the element of the dual basis $e^{pq}$, so the action of any matrix $b\in\sln$ on partial derivations on $S(\sln)$ now coincides with the dual action of $b$ on $\sln^*$. So for any $b\in\mathfrak b_n$ we can compute:
\[
\left(\frac{\partial}{\partial x_{k,n-k+1}}\right)^b=-(b_{kk}-b_{n-k+1,n-k+1})\frac{\partial}{\partial x_{k,n-k+1}}+\sum_{i=1}^{k-1}\lambda_i\frac{\partial}{\partial x_{i,n-k+1}}+\sum_{j=1}^{k-1}\mu_j\frac{\partial}{\partial x_{k,n-k+j+1}}
\]
for certain coefficients $\lambda_i,\,\mu_j$. However, since for any $p,q,r,s,t,u$ and $i$
\[
\frac{\partial^2E_i}{\partial x_{pr}\partial x_{ps}}=0=\frac{\partial^2E_i}{\partial x_{uq}\partial x_{tq}},
\]
we have
\begin{equation}\label{eq:equivarprop}
E_{k,i}^b=-\left(\sum_{r=1}^{k}(b_{kk}-b_{n-k+1,n-k+1}) \right)E_{k,i}.
\end{equation}
Here the expression $-\sum_{r=1}^{k}(b_{kk}-b_{n-k+1,n-k+1})=c_k(b)$ is the character on $\mathfrak b_n$. It is evident, that this character is common for all $E_{k,1},\,E_{k,2},\dots,E_{k,n}$; in effect, it is induced from the operator $D_k$, common for all these functions.
\end{proof}
Since the character $c_k$ does not depend on the second index, we can obtain invariants of the action simply by taking the ratios $E_{k,\ell}/E_{k,m}$ for different $\ell$ and $m$. Then applying the Lemma \ref{lem:aks} gives the Deift \textit{chopping integrals} of the Toda system (the signs that appear in the formula \eqref{eq:altchop1} cancel out): these are rational functions on $\bn^*\cong Symm_n(\R)$, they commute with each other and with the Toda flow, since the latter corresponds to the function $\frac12Tr(L^2)$ on $\sln^*$, which is in the Poisson center of $S(\sln)$. Applying the AKS construction now gives the commutative family of the rational first integrals of the Toda system.

\subsection{Chernyakov-Sorin family}\label{sect:CSRepr}
There exist many other approaches to the construction of the integrals of the full symmetric Toda system. They give the integrals, slightly different from the ones defined by Deift's chopping procedure, which commute with the Toda flow, but are in general not in involution.

For example, see \cite{CS}, instead of the matrix $L-\lambda\mathbbm 1$, one can consider the matrices $L^i$, so that $K_i(L)=\det(L^i)=(\det L)^i$. Then $D_k(K_i(L))$ does not belong to the subalgebras $S(\mathfrak p^k_n)\subset S(\sln)$ if $i>1$, so it is not true that $D_k(K_i)$ Poisson commute with each other or with the Toda flow. However, we can do the following trick: put
\[
K_{k,i}(L)=\det((L^i)^{(k)}),
\]
where as before for $A\in Mat_n(\R)$ we denote by $A^{(k)}=P^k(A)$ its submatrix, spanned by the last $k$ rows and the first $k$ columns. Then the functions $K_{k,i}(L)$ also are semiinvariants of the Toda flow, and their ratios $I^{cs}_{k,i}(L)=\frac{K_{k,i}(L)}{K_{k1}(L)}$ are genuine integrals for the principal Toda flow. The proof of this statement can be obtained by consideration of Pl\"ucker coordinates, see \cite{CS}.

In effect the functions $\frac{K_{k,i}(L)}{K_{k,1}(L)}$ are invariant not just of the Toda flow, but they are invariant with respect to the adjoint action by \bn; this can be proved analogously to the above reasonings or also c.f. the arguments of section 2.4 of \cite{T1}. We provide here another demonstration based on representation theory methods, similar to the construction of \cite{CSS19}.
\begin{lemma}
\label{lem_binv} Let $\mathfrak{K}$ denotes the subalgebra in $Q(\sln)$, generated by the ratios $\frac{K_{k,i}(L)}{K_{k1}(L)}$; then
\[
\mathfrak{K}\subset Q(\mathfrak{sl}_n)^{\mathfrak{b}}.
\]
\end{lemma}
\begin{proof}
Let us consider a finite-dimensional irreducible representation $\rho: SL_n(\R) \rightarrow End(V)$ with $v_+$ - the highest weight vector and $v_-$ - the lowest weight vector. These vectors are characterised by the conditions:
\begin{equation}\label{eq:nilpotact}
\rho(\mathfrak{n}_+) v_+=0;\qquad \rho(\mathfrak{n}_- )v_-=0
\end{equation}
where $\mathfrak{n}_{\pm}$ are the corresponding upper and lower nilpotent triangular subalgebras and we denote by the same symbol $\rho$ the representation of the Lie algebra \sln, induced by $\rho$. Let us now choose the $SO_n(\R)$-invariant inner product on $V$, denoted $\langle,\rangle$. Then the following conditions are satisfied:
\bea
\langle\rho(X)v,w\rangle=\langle v,\rho(X^T)w\rangle.\nn
\eea
We now consider a special subset of matrix elements in the representation
\bea
F^\rho(X)=\langle\rho(X) v_+,v_-\rangle.
\eea
These functions are invariant with respect to the adjoint action of the unipotent subgroup $N_+$ of upper triangular matrices with $1$ on diagonal of $SL_n(\R)$; here we define the adjoint action as
\bea
Ad_g X=g^{-1} X g.\nn
\eea
Indeed, since the conditions on the action of $\mathfrak n_\pm$ on $v_\pm$ (see \eqref{eq:nilpotact}) integrate to the equality $\rho(g)(v_\pm)=v_\pm$ for $g\in N_\pm$, we have
\bea
Ad^*_g(F^\rho)(X)=\langle\rho(g^{-1}Xg)v_+,v_-\rangle=\langle\rho(X)\rho(g)v_+,\rho((g^{-1})^T)v_-\rangle=\langle\rho(X)v_+,v_-\rangle.\nn
\eea
We can now reinterpret the Chernyakov-Sorin integrals in this framework: it turns out that the semiinvariants $K_{k,i}(L)$ are related to the functions $F^\rho$ with $\rho$ being the highest weight irreducible subrepresentation in $\wedge^k(S^j V_0)$ where $V_0$ is the fundamental representation with the highest weight $(1,0,\ldots,0)$.  First of all let us consider the exterior power of the fundamental representation $\wedge^k(V_0)$; we denote it by $\rho_k$. This representation has the weight basis given by $e_{i_1}\wedge\ldots\wedge e_{i_k},\,1\le i_1<i_2<\dots<i_k\le n$ where $e_i$ are vectors from the weight basis in $V_0$ with $v_+=v_1$ and $v_-=v_n$. Then the highest weight vector in $\wedge^k(V_0)$ is $e_1\wedge\ldots\wedge e_k$ and the lowest one is $e_{n-k+1}\wedge\ldots\wedge e_n$. The left lowest minor $\Delta_k (X)$ is the matrix element of $X^{\wedge k}$ corresponing to $F^{\rho_k}$. Then as it is easy to see that $\Delta_k(X^i)$ is equal to $F^{\rho_k}(X^i)$.

Let us demonstrate that for all $i$ and fixed $k$ the functions $K_{k,i}(L)$ are semi-invariants with the same character $c_k$ with respect to the adjoint action of the Cartan subgroup $T_n$ (the subgroup of diagonal matrices in $SL_n(\R)$. Indeed, the action of $T_n$ on $K_{k,i}(L)$ is induced from its action on $v_{\pm}$ which is the highest and the lowest weights in $\wedge^k(V_0)$; on the other hand for all $t\in T_n$ we have
\bea
\rho_k(t) v_\pm ={c_k}_\pm(t) v_\pm,\nn
\eea
where ${c_k}_\pm$ are the weights of the vectors $v_\pm$ in the representation $V_0$. Hence
\bea
Ad_t^*(F^{\rho_k})(X)=\langle\rho_k(X)\rho_k(t) v_+,\rho_k(t^{-1})v_-\rangle=\frac{{c_k}_+}{{c_k}_-} \langle\rho_k(X) v_+,v_-\rangle=\frac{{c_k}_+}{{c_k}_-} F^{\rho_k}(X).\nn
\eea
As one readily sees, we get the same character $c_k=\frac{{c_k}_+}{{c_k}_-}$ for the functions $\Delta_k(X^i)$ for all $i$:
\[
\begin{aligned}
Ad_t^*(F^{\rho_k})(X^i)&=\langle(\rho_k(t^{-1})\rho_k(X)\rho_k(t))^i v_+,v_-\rangle=\langle\rho_k(t^{-1})\rho_k(X^i)\rho_k(t) v_+,v_-\rangle\\
&=\frac{{c_k}_+}{{c_k}_-}\langle\rho_k(X^i) v_+,v_-\rangle=\frac{{c_k}_+}{{c_k}_-} F^{\rho_k}(X^i).
\end{aligned}
\]
\end{proof}

\begin{remark}\rm
The Toda system is superintegrable in the sense of Nekhoroshev (see for instance \cite{CS}). It is closely related to the classical question of describing the ring of rational functions on $G/AdB$, see for example \cite{RS}. Such functions are ratios of semi-invariants with the same weights. On the other hand, the semi-invariants are given by the highest vectors in the regular representation in the group algebra $C[G]=\bigoplus_{\lambda} V_{\lambda}\otimes V_{\lambda^*}$ (due to the Peter-Weyl theorem \cite{Hall}). The latter can be obtained by using the Littlewood-Richardson rule for representations $V_{\lambda}\otimes V_{\lambda^*}$.
\end{remark}



\section{Lax representations and solutions}\label{sect:laxqr}
One of the important constructions that allow one find solutions of the integrable systems is the \textit{Lax representation} of the equation; earlier we have described the canonical Lax representation of the Toda system, see equation \eqref{eq:laxtoda}. In this section we are going to describe the Lax representations for the DLNT (chopping) integrals. We begin with considering the examples of such Lax representations and describe the way they lead to the solution of the corresponding flows in terms of the QR-decomposition.
\subsection{Example: the Lax representation of a DLNT flow for the Toda system on $\mathfrak{sl}_6$}\label{sect:sl6}
Let us consider the $n=6$ case. We are going to give the Lax representation of the flow corresponding to the Chernyakov-Sorin hamiltonian
\bea
I^{cs}_{1,2}(L)=\frac {(L^3)_{16}}{a_{16}},\nn
\eea
where $(L^3)_{ij}$ are the matrix elements of $L^3=L \cdot L \cdot L$. Recall that the Lax matrix $L\in Symm_6$ is the symmetric matrix with entries $a_{ij}$, so that $a_{ij}=a_{ji}$.

The integral $I^{cs}_{1,2}(L)$ and the integral
$$I_{1,2}(L) = \frac{E_{1,2}(L)}{E_{1,0}(L)}$$
obtained by chopping procedure are functionally dependent as the integrals of the Toda flow (see Appendix A). The direct computations show (see Appendix B) that the dynamics associated with this function has the following Lax representation
\beq{Flows}
\mathcal{T}_{I^{cs}_{1,2}}(L)=\mathcal{T}_{\frac12TrL^{2}}(L)+\mathcal{T}_{I_{1,2}}(L),
\eq
As one readily sees the first term in this equality is just the Toda vector field. In coordinate form, we have the following equality
$$\frac{dL}{dt_{I^{cs}_{1,2}}}=[M(L),L]+[M_{I_{1,2}}(L),L],$$
where we use the notation from formula \eqref{eq:moper}, $M(L)=L_+-L_-$ and $M_{I_{1,2}}(L)$ being the matrix of the form
\bea
M_{I_{1,2}}(L)=\frac 1 {a_{16}} \left(
\begin{array}{rrrrrr}
0 & 0 & 0 & 0 & 0 & 0\\
0 & 0 & L^{36}_{12} & L^{46}_{12} & L^{56}_{12} & 0\\
0 & - L^{36}_{12} & 0 & L^{46}_{13} & L^{56}_{13} & 0\\
0 & - L^{46}_{12} & - L^{46}_{13} & 0 & L^{56}_{14} & 0\\
0 & - L^{56}_{12} & - L^{56}_{13} & - L^{56}_{14} & 0 & 0\\
0 & 0 & 0 & 0 & 0 & 0
\end{array}
\right)= \frac 1 {a_{16}} M_h.
\eea
Here and further in this paper for a matrix $A$ we denote by $A^{q_1\dots q_k}_{p_1\dots p_k}$ the determinant of the submatrix in $A$, spanned by the intersections the rows $q_1,\dots,q_k$ and the columns $p_1,\dots,p_k$.
The denominator matrix $M_h$ in the formula for $M_{I_{1,2}}(L)$ has the following representation-theoretic interpretation. Let $V=\R^6$ be the fundamental representation of $SL_6(\R).$ We denote by $\rho$ the representation map $\rho: SL_6(\R)\rightarrow End(V)$. Then denote by $\wedge^2\rho$ the map
\bea
\wedge^2\rho: SL_6(\R)\rightarrow End(\wedge^2 V)\nn
\eea
which acts as
\bea
\wedge^2 \rho(X): v_1\wedge v_2\mapsto \rho(X)v_1\wedge \rho(X) v_2.\nn
\eea
Let us recall that the matrix element $a_{16}$ has the natural representation-theoretic interpretation
\bea
a_{16}(L)=\langle\rho(L) v_-,v_+\rangle\nn
\eea
Now the Lax matrix $M_h$ as a whole has a similar interpretation. Consider the operator $\wedge^2 \rho(L)$ as an element of
\[
End(V\wedge V)=(V\wedge V)^*\otimes(V\wedge V)\subset (V^*\otimes V^*) \otimes (V\otimes V).
\]
Now there is a natural operation
\[
j:(V^*\otimes V^*) \otimes (V\otimes V)\to V^*\otimes V=End(V),
\]
given by
\bea
j((f\otimes g)\otimes(x\otimes y))=g(v_+) \langle x, v_-\rangle f\otimes y\in V\otimes V^*.\nn
\eea
Then one readily sees that
\bea
M_h= M(j(\rho^2(L))),
\eea
where $M$ is the projector \eqref{eq:moper}, see remark \ref{rem:moper1}.

\subsection{Equivariant tensor powers}
Building on the observations we just made, we get the following construction. Let $\{ e_i \}$ be the standard basis and let $\langle,\rangle$ be the standard scalar product in $\R^n$. Consider the homogenous polynomial function of the Lax matrix $L=(a_{ij})$ of degree $k-1$, given by the following determinant:
\bea
C_k(L)=\left|
\begin{array}{ccc}
\langle Le_{n-k+2}, e_1\rangle & \ldots & \langle L e_n, e_1\rangle\\
\vdots &\ddots & \vdots \\
\langle Le_{n-k+2}, e_{k-1}\rangle & \ldots & \langle L e_n, e_{k-1}\rangle
\end{array}
\right|.\nn
\eea
Here we assume that $n-k+2\le n$, so that $C_1(L)=1$; in effect, this function can be defined for any matrix $L\in\sln$, not necessarily for the symmetric matrix. In a similar way we define a rational matrix-valued function $\tau_k:\sln\to Mat_n(\R)\cong\gl_n$: for any $v\in\R^n$ we put
\bea
\tau_k(L)(v)=\frac{1}{C_k(L)}
\left|
\begin{array}{cccc}
Lv & L e_{n-k+2} & \ldots  & L e_n \\
\langle L v, e_1\rangle & \langle Le_{n-k+2}, e_1\rangle & \ldots & \langle L e_n, e_1\rangle \\
\vdots & & & \vdots \\
\langle L v, e_{k-1}\rangle & \langle Le_{n-k+2}, e_{k-1}\rangle & \ldots & \langle L e_n, e_{k-1}\rangle
\end{array}
\right|
\label{tau}
\eea
for any $v\in\R^n$. Observe that by definition
\bea
\tau_1(L)=L.\nn
\eea
As an example, let us compute the matrix $\tau_2(L)$: first, we have by definition
\[
C_2(L)=\langle Le_n,e_1\rangle =a_{1n}.
\]
Also
\[
\tau_2(L)(v)=\langle L(e_n),e_1\rangle Lv-\langle L v, e_1\rangle Le_n=a_{1n}Lv-\langle L v, e_1\rangle a_{jn}e_j,
\]
so that
\[
\tau_2(L)(e_i)=(a_{1n}a_{ji}-a_{1i}a_{jn})e_j.
\]
In other words: $\tau_2(L)_{ji}=-L^{1j}_{in}$, in particular $\tau_2(L)_{ji}=0$ if $j=1$ or $i=n$. Similarly
\begin{equation}\label{eq:matrixcoef1}
\tau_k(L)_{ji}=(-1)^{k-1}L^{1,2,\dots,k-1,j}_{i,n-k+2,\dots,n},
\end{equation}
where the corresponding determinant vanishes, if the indices are repeating.

The following lemma describes one of the important properties of $\tau_k(L)$:
\begin{lemma}
The function $\tau_k(L)$ is covariant with respect to the adjoint action of the lower triangular group $B^-_n$, i.e.
\bea
\tau_k(bLb^{-1})=b \tau_k(L) b^{-1}.\nn
\eea
\end{lemma}
\begin{proof}
Let us consider separately the Cartan and unipotent parts of the group $B_-$. First we consider a unipotent lower-triangular element $n$. It is characterised by the property:
\bea
n^{-1}(e_n)&=&e_n;\nn\\
n^{-1}(e_{n-1})&=&e_{n-1}+\alpha^n_{n-1} e_{n};\nn\\
\ldots&&\nn\\
n^{-1}(e_j)&=&e_j+\alpha^{j+1}_j e_{j+1}+\ldots+\alpha^n_j e_{n}
\label{rel-ninv}
\eea
for all $j$. Similarly,
\bea
n^T(e_1)&=&e_1;\nn\\
n^T(e_{2})&=&e_{2}+\beta^2_{1} e_{1};\nn\\
\ldots&&\nn\\
n^T(e_j)&=&e_j+\beta^{j}_{j-1} e_{j-1}+\ldots+\beta^j_1e_{1}
\label{rel-ntr}
\eea
for all $j$. This is sufficient to demonstrate that $C_k(L)$ is invariant:
\bea
&&C_k(nLn^{-1})=\left|
\begin{array}{cccc}
\langle nLn^{-1}e_{n-k+2}, e_1\rangle & \ldots & \langle nLn^{-1} e_{n-1}, e_1\rangle & \langle nLn^{-1} e_n, e_1\rangle\\
\vdots & & & \vdots \\
\langle nLn^{-1}e_{n-k+2}, e_{k-1}\rangle & \ldots & \langle nLn^{-1} e_{n-1}, e_{k-1}\rangle & \langle nLn^{-1} e_n, e_{k-1}\rangle
\end{array}
\right|\nn\\
&=&\left|
\begin{array}{cccc}
\langle nL(e_{n-k+2}+\ldots +\alpha_{n-k+2}^n e_n), e_1\rangle & \ldots & \langle nL(e_{n-1}+\alpha^n_{n-1} e_{n}), e_1\rangle & \langle nL e_n, e_1\rangle\\
\vdots & & & \vdots \\
\langle nL(e_{n-k+2}+\ldots +\alpha_{n-k+2}^n e_n), e_{k-1}\rangle & \ldots & \langle nL(e_{n-1}+\alpha^n_{n-1} e_{n}), e_{k-1}\rangle & \langle nL e_n, e_{k-1}\rangle
\end{array}
\right|\nn\\
&=&\left|
\begin{array}{cccc}
\langle nLe_{n-k+2}, e_1\rangle & \ldots & \langle nLe_{n-1}, e_1\rangle & \langle nL e_n, e_1\rangle\\
\vdots & & & \vdots \\
\langle nLe_{n-k+2}, e_{k-1}\rangle & \ldots & \langle nLe_{n-1}, e_{k-1}\rangle & \langle nL e_n, e_{k-1}\rangle
\end{array}
\right|.\nn
\eea
The first equality is due to (\ref{rel-ninv}), the second one is the consequence of the basic determinant property: determinant of a matrix is invariant under the linear transformations with columns. Using the same arguments with the rows of this matrix and relations (\ref{rel-ntr}) we get
\bea
C_k(nLn^{-1})&=&
\left|
\begin{array}{cccc}
\langle Le_{n-k+2}, e_1\rangle & \ldots & \langle Le_{n-1}, e_1\rangle & \langle L e_n, e_1\rangle\\
\vdots & & & \vdots \\
\langle Le_{n-k+2}, e_{k-1}\rangle & \ldots & \langle Le_{n-1}, e_{k-1}\rangle & \langle L e_n, e_{k-1}\rangle
\end{array}
\right|=C_k(L).\nn
\eea
The same argument is valid for the numerator of the formula (\ref{tau}). Now we consider the action of the diagonal matrices on $\tau_k(L)$: it is easy to see that both the numerator and the denominator of this expression change in the same manner under the action of a diagonal matrix $d=diag(d_1,\ldots,d_n)$; in effect they are multiplied by the expression
\bea
C_k(dLd^{-1})=\frac {d_1 \ldots d_k}{d_{n-k+1}\ldots d_n}.\nn
\eea
\end{proof}
\begin{remark}\label{rem:charprop}\rm
In effect, the similar covariance condition holds not only for the group $B^-_n$, but for any group inside the parabolic subgroup $\mathbf{P}_k\subseteq SL_n(\R)$, where $\mathbf{P}_K$ is the group of invertible matrices inside the union of the groups $B_-\bigcup P_n^k$ (here $P_n^k$ is the Lie group with Lie algebra $\mathfrak p_n^k$); in other words, this is the group of invertible matrices that have the form
\setcounter{MaxMatrixCols}{15}
 \[{\tiny
{\footnotesize X\!=\!}\begin{pmatrix}
x_{11}      & 0              & \dots        & 0          & 0                & \dots         & 0                & 0         &0    & \dots & 0\\
x_{21}      & x_{22}      & \dots        & 0          & 0                 & \dots         & 0                & 0         &0    & \dots & 0\\
\vdots      & \vdots       & \ddots &\vdots         & \vdots         & \ddots &\vdots         &  \vdots        &\vdots  & \ddots &\vdots\\
x_{k1}      & x_{k2}      & \dots   & x_{kk}        & 0                 & \dots        & 0                &  0        &0    & \dots & 0\\
x_{k+1,1} & x_{k+1,2} &\dots   & x_{k+1,k}    & x_{k+1,k+1}& \dots & x_{k+1,n-k} & 0        &0 & \dots & 0\\
x_{k+2,1} & x_{k+2,2} &\dots    & x_{k+2,k}    & x_{k+2,k+1}& \dots & x_{k+2,n-k} & 0         &0 & \dots & 0\\
\vdots      & \vdots       & \ddots & \vdots         & \vdots         & \ddots &\vdots         & \vdots  &\vdots &\ddots&\vdots\\
x_{n-k,1}     & x_{n-k,2}      &\dots     & x_{n-k,k}     & x_{n-k,k+1}      & \dots   & x_{n-k,n-k}     & 0                         & 0      & \dots & 0\\
x_{n-k+1,1} & x_{n-k+1,2} &\dots     & x_{n-k+1,k} & x_{n-k+1,k+1}  &\dots     & x_{n-k+1,n-k} & x_{n-k+1,n-k+1} & 0      & \dots & 0\\
x_{n-k+2,1} & x_{n-k+2,2} &\dots     & x_{n-k+2,k} & x_{n-k+2,k+1}  &\dots     & x_{n-k+2,n-k} & x_{n-k+2,n-k+1} & x_{n-k+2,n-k+2}      & \dots & 0\\
\vdots         & \vdots          &\ddots    & \vdots         & \vdots              &\ddots    &\vdots              &\vdots                &\vdots &\ddots& \vdots\\
x_{n1}        & x_{n2}          &\dots      & x_{nk}         & x_{n,k+1}      & \dots &x_{n,n-k}             &x_{n,n-k+1}       & x_{n,n-k+2}&\dots   &x_{nn}
\end{pmatrix}
} \]
For instance this is true for the matrices $O$ of orthogonal transformations of $\R^n$, for which
\[
O(e_1)=e_1,\dots,O(e_k)=e_k,\ \mbox{and}\ O(e_{n-k+1})=e_{n-k+1},\dots,O(e_{n})=e_n.
\]
Similar properties hold for the adjoint action of the corresponding Lie algebras, in particular, for any element $X$ from $\mathfrak p_n^k+\mathfrak b^-_n$: we have
\begin{equation}\label{eq:ddiff}
[X,\tau_k(L)]=\frac{d}{dt}_{|_{t=0}}\tau_k(e^{tX}Le^{-tX}).
\end{equation}
We shall denote the right hand side of this formula by $\tau_k([X,\dot L])$.
\end{remark}
\noindent Let us conclude this section by observing that due to the definition of $\tau_k(L)$ we have the following obvious property
\begin{lemma}
\label{lem-null}
\bea
\tau_k(L)e_n&=&\tau_k(L) e_{k-1}=\ldots=\tau_k(L) e_{n-k+1}=0;\nn\\
\langle \tau_k(L) v, e_1\rangle&=&\langle\tau_k(L) v, e_2\rangle=\ldots=\langle\tau_k(L) v, e_{k-1}\rangle=0.\nn
\eea
\end{lemma}
\subsection{QR-decomposition}\label{sect:qrdec1}
Let $\tau_k(L)$ be the matrix-valued function of $L$, as we explained above; we define the field $T_k(L)$ on $Symm_n$ by the formula
\[
T_k:Symm_n(\R)\to Symm_n(\R),\ T_k(L)=[M(\tau_k(L)),L].
\]
Below (see section \ref{sect:hamilttk}) we will show that $T_k(L)$ are in fact the Hamiltonian fields of certain DLNT integrals, namely $T_k(L)$ is the Hamilton field of $H_{k,n-1}(L)$. This will automatically mean that these fields commute with each other for all possible pairs of indices. For us one of the main advantages that come from the working with these fields is the fact that their flows have solutions, given by the QR-decomposition, just like the original Toda flow; namely, the following is true:
\begin{Th}\label{th:1}
Let us consider the flow of the vector field $T_k(L)$ on $Symm_n$:
\bea
\label{Lax}
\frac {dL}{d t_{k}}=[M(\tau_k(L)),L].
\eea
Then the solution $L=L(t_k)$ of this equation can be obtained by the following procedure: let $L_0$ is the initial point in the phase space of symmetric matrices. Let us consider the $QR$ decomposition of the exponential:
\bea
\label{QR}
\exp( t \tau_k(L_0))=R(t) Q(t)
\eea
where $t=t_{k}.$ Then
\bea
L(t)=Q(t) L_0 Q^{-1}(t)\nn
\eea
is a solution for (\ref{Lax}) with the initial value $L(0)=L_0$.
\end{Th}
\begin{proof} The proof is done by a suitable modification of the arguments used for the isospectral flows. Let us differentiate (\ref{QR}) by $t$:
\bea
R Q\tau_k(L_0)=\dot{R}Q+R\dot{Q}. \nn
\eea
This is equivalent to the following expression
\bea
Q\tau_k(L_0) Q^{-1} = R^{-1}\dot{R} +\dot{Q}Q^{-1}.\nn
\eea
Here $R^{-1}\dot{R}$ is a lower triangular, and $\dot{Q}Q^{-1}$ is orthogonal matrix; so
\[
M(Q\tau_k(L_0) Q^{-1})=\dot{Q}Q^{-1}.
\]
On the other hand
\bea
\frac {dL}{dt}=\dot{Q} L_0 Q^{-1}-QL_0Q^{-1}\dot{Q} Q^{-1}=[\dot{Q} Q^{-1},Q L_0 Q^{-1} ].\nn
\eea
It remains to demonstrate that $\dot{Q}Q^{-1}$ is exactly $M(\tau_k(L(t)))$. First we remark that the orthogonal matrix from the $QR$-decomposition is of the very particular form. Indeed, the element $\eta(t)=\exp( t \tau_k(L_0))$ is such that
\bea
\eta(t) e_i= e_i \qquad i=1,\ldots, k-1, n-k+1,\ldots, n.\nn
\eea
Clearly the same property holds for the matrix $Q(t)$. Hence, by the characteristic property of the map $\tau_k$, see remark \ref{rem:charprop}
\bea
Q\tau_k(L_0) Q^{-1}=\tau_k(Q L_0 Q^{-1}).\nn
\eea
Summing up all these observations we obtain
\bea
 \dot{Q}Q^{-1}=M(Q\tau_k(L_0) Q^{-1})=M(\tau_k(Q L_0 Q^{-1}))=M(\tau_k(L(t))).
\eea
\end{proof}

\subsection{Hamiltonian description}\label{sect:hamilttk}
Let us now show that $T_k(L)$ is in fact equal to the Hamilton field of a Deift chopping Hamiltonian. To this end recall that if $E_{k,n-i}(L)$ are the coefficients of the partial characteristic polynomial $\chi_k(\lambda)=det((L-\lambda \mathbbm 1)^{(k)})$ as defined in (\ref{Pkn}), (also c.f. \eqref{eq:altchop1}) i.e. $E_{k,n-i}(L)$ is the coefficient at $\lambda^{n-2k-i}$, then the chopping Hamiltonians are defined as the ratios
\bea
I_{k,n-i}(L)=\frac{E_{k,n-i}(L)}{E_{k,0}(L)}.\nn
\eea
As we showed earlier, they constitute a Poisson-commutative family. Then the following is true
\begin{Th}\label{theo:tauk}
The vector field $T_k(L)=[M(\tau_{k}(L)),L]$ is Hamiltonian with respect to the Poisson structure on $Symm_n$ induced from the Kirillov-Kostant bracket on $C^\infty(\mathfrak b_-)$ and the Hamiltonian $I_{k,n-1}$.
\end{Th}
\begin{proof}
We begin with recalling (see section \ref{sect:chop1}) that the Hamilton field of a smooth function $f:Symm_n\to\R$ with respect to the Poisson structure, induced from $\mathfrak b_n$ is equal to:
\[
X_f=[M(\nabla f),L]
\]
where $M$ is the projector to $\mathfrak{so}_n$ that we have used earlier and for any smooth function on $\mathfrak{sl}_n$, $\nabla f$ denotes the matrix of partial derivatives of $f$: if $\partial_{ij}=\frac{\partial}{\partial x_{ij}}$, then (see formula \eqref{eq:nabladef})
\[
\nabla f=\begin{pmatrix}
\partial_{11}f &\dots &\partial_{1n}f\\
\vdots &\ddots &\vdots\\
\partial_{n1}f &\dots &\partial_{nn}f
\end{pmatrix}.
\]
Observe that on $Symm_n$ we have $x_{ij}=x_{ji}$ and this matrix is in effect symmetric: $\partial_{ij}(f)=\partial_{ji}f$. Now in order to show that $T_k$ coincides with the Hamilton field of $I_{k,n-1}$, we need to show that $M(\tau_k(L))$ coincides with $M(\nabla I_{k,n-1}(L))$ for all symmetric $L$. To this end let us consider the upper-triangular part of the expression
\bea
d\left(\frac{\chi_k(\lambda)}{E_{k,0}(L)}\right).\nn
\eea
That is the coefficients at $d(a_{ij})$ with $i<j$ of the differential. Such a coefficient for the numerator is a minor of $(L-\lambda \mathbbm 1)^{(k)}$ with excluded $i$-th row and $j$-th column. This is a polynomial of $\lambda$ of degree $n-2k-2$. Its highest term is equal to the minor of $(L-\lambda \mathbbm 1)^{(k)}$ with conversely $i$-th column and $j$-th row removed, since the matrix $L$ is symmetric; taking the coefficient at $\lambda^{n-2k-1}$ of this expression we get just $L_{i,n-k+2,\ldots,n}^{1,\ldots,k-1,j}$ which is exactly the upper-tringular matrix element for the matrix $\tau_k(L)$, see the formula \eqref{eq:matrixcoef1}. We finally recall again that the differential of $E_{k,n-2k}(L)$ does not enter the upper-triangular part of the differential. Hence we get
\bea
M(\nabla I_{k,n-1}(L))=M(\tau_k(L)).
\eea
\end{proof}

\subsection{Higher chopping integrals and QR-decomposition method}\label{sect:chopqr1}
%
%
%
%
%

In the previous sections we described the integration procedure of the Hamiltonian flow, induced from the chopping integral $I_{k,n-1}(L)$; it was based on an explicit formula for the corresponding Hamiltonian vector field $T_k(L)$. Now in this section we will show that the methods used to integrate the fields $T_k(L)$ can be applied to all the chopping integrals. Namely, using the explicit formula for the functions $E_{k,i}(L)$ based on the application of the differential operator $D_k$, we can now describe the Hamiltonian fields, associated with the chopping integrals and show that a variant of QR-decomposition method of solving the corresponding flow equations works in this case too.

To this end we first of all observe that the matrix $\nabla\det L$ (for not necessarily symmetric $L$, i.e. we do not assume $x_{ij}=x_{ji}$) is given by the following formula, similar to the definition of $\tau_k(L)$:
\begin{equation}\label{eq:nablall}
\nabla\det L(v)=\begin{vmatrix}
0 & e_1 &\dots & e_n\\
v_1& x_{11}&\dots & x_{1n}\\
\vdots &\vdots &\ddots&\vdots\\
v_n & x_{n1}&\dots &x_{nn}
\end{vmatrix}
\end{equation}
where $v=v_1e_1+\dots+v_ne_n$ is the representation of $v\in\R^n$ in basis $e_1,\dots,e_n$. To see this, it is enough to recall the formula $\partial_{ji}(\det L)=\widehat L_i^j$ and so
\begin{equation}\label{eq:nablalll}
 \nabla(\det L)(e_i)=\widehat L_i^je_j.
\end{equation}
Here $\widehat L_i^j$ denotes the algebraic complement of the matrix $L$ at $ij$-th place, i.e., using the notation from the section \ref{sect:sl6}, we have
\[
\widehat L_i^j=(-1)^{i+j}L^{1,\dots,\widehat j,\dots,n}_{1,\dots,\widehat i,\dots,n},
\]
where as usual the hats above indices signify their absence from the formula. Now the comparison of \eqref{eq:nablall} and \eqref{eq:nablalll} proves the result.

It follows that the matrices $\nabla E_{k,i}(L)$ can be obtained up to a sign by pointwise application of the operator $D_k$ to the matrix coefficients of the matrix standing at $\lambda^{n-i}$ in the expression $\nabla\det(L-\lambda\mathbbm1)$,
\begin{equation}\label{eq:nablaeki}
\sum_i (-1)^{n-i}\lambda^{n-i}(\nabla E_{k,i}(L))(v)=D_k\begin{vmatrix}
0 & e_1 &\dots & e_n\\
v_1& x_{11}-\lambda&\dots & x_{1n}\\
\vdots &\vdots &\ddots&\vdots\\
v_n & x_{n1}&\dots &x_{nn}-\lambda
\end{vmatrix}
\end{equation}
This follows from the fact that both $\nabla$ and $D_k$ are differential operators with constant coefficients and hence they commute with each other. Also remark that since $E_{k,0}(L)$ does not depend on the coefficients $x_{ij}$ with $i>j$, matrix $\nabla E_{k,0}(L)$ is strictly lower-triangular matrix, and hence it is killed by the projection $M$. So we have
\[
M\left(\nabla\left(\frac{E_{k,i}(L)}{E_{k,0}(L)}\right)\right)=\frac{1}{E_{k,0}(L)}M(\nabla E_{k,i}(L)).
\]
We can finally show that the QR-decomposition method can be applied to the flows of the functions
$$I_{k,i}(L)=\frac{E_{k,i}(L)}{E_{k,0}(L)}.$$
To this end we observe that the matrix
$$\mathbf E_{k,i}(L)=\nabla\left(\frac{E_{k,i}(L)}{E_{k,0}(L)}\right)$$
 verifies the same characteristic condition as $\tau_k(L)$:
\begin{Prop}
\begin{equation}\label{eq:charproperty}
g\mathbf E_{k,i}(L)g^{-1}=\mathbf E_{k,i}(gLg^{-1})
\end{equation}
for all orthogonal matrices $g$ in the union $B_-\bigcup P_n^k$.
\end{Prop}
\begin{proof} Since $E_i(L)=E_i(gLg^{-1})$ for all invertible matrices, we have (compare with the proof of the lemma \ref{lem:key2})
\[
(\nabla E_i)(gLg^{-1})=\nabla^g(E_i(gLg^{-1}))=\nabla^g E_i(L).
\]
Here $g$ acts on the matrix-valued differential operator $\nabla$ by the conjugate action, i.e.
\[
\nabla^g=\rho(g)\nabla \rho(g)^{-1},\ \mbox{where}\ \rho(g)=(g^T)^{-1}.
\]
So we have
\[
(\nabla E_i)(gLg^{-1})=\rho(g)\left(\nabla E_i(L)\right)\rho(g)^{-1}
\]
for all invertible matrices $g$, and since $\rho(g)=g$ for all orthogonal matrices, we have
\[
(\nabla E_i)(gLg^{-1})=g\left(\nabla E_i(L)\right)g^{-1}
\]
for all orthogonal matrices $g$. On the other hand, since $\nabla$ commutes with $D_k$ we have, just like in the equation \eqref{eq:equivarprop}:
\[
g\nabla E_{k,i}(L)g^{-1}=c_k(g)(\nabla E_{k,i})(gLg^{-1})
\]
for any orthogonal $g\in B_-\bigcup P_n^k$. Since the same character $c_k(b)$ pops out in the denominator $E_{k,0}(L)$ of the expression for $\mathbf E_{k,i}(L)$, the equality \eqref{eq:charproperty} holds.
\end{proof}
\begin{remark}\rm
The formula \eqref{eq:nablall} and its corollary although quite handy is not in fact indispensable for proving the characteristic equation \eqref{eq:charproperty}: it is rather the fact that we can define $E_{k,i}$ with the help of the operator $D_k$, that plays the crucial role in the reasoning.
\end{remark}
Now we can show that the QR-decomposition method gives solutions of the higher Deift flows: for the initial condition $L_0\in Symm_n$ we consider the QR-decomposition of the following exponent
\[
\exp(t\mathbf{E}_{k,i}(L_0))=R(t)Q(t),
\]
then $L(t)=Q(t)L_0Q(t)^T$ solves the equation
\[
\dot L=[M(\mathbf{E}_{k,i}(L)),L].
\]
The reasoning is just like in the proof of the theorem \ref{th:1} above.

\section{Conclusions and remarks}
I this section we gathered few results and observations that do not fit into the main body of the text, but we believe are quite interesting on their own. One may say that these are curious observations, that can lead to further investigations.
\subsection{Commutation relation}
It would be interesting to find a geometric mechanism underlying the commutativity of Deift integrals: we know that  the fields $T_k(L)$ are the Hamiltonian fields of $I_{k,n-1}(L)$, so they should commute with each other by general theory; but one can ask if it is possible to prove the commutation relations independently of the identification of $T_k(L)$ with Hamiltonian fields. This might help one to look for other commutative families and symmetries of the Toda system, c.f. \cite{CSS19} for example.

It turns out that up to a certain degree this can be done by direct computation without the AKS theory. To this end recall that $M:\sln\to\mathfrak{so}_n$ is the natural projection along the subalgebra $\mathfrak b_-$ associated with the direct sum decomposition \eqref{eq:triangledec} and $T_k(L)$ are the vector fields, given by equation
\[
T_k(L)=[M(\tau_k(L)),L]
\]
In particular
\[
T_1(L)=[M(L),L]
\]
is the Toda flow field.
Then the following is true: the fields $T_k(L)$ commute with the Toda flow, i.e. $[T_1,T_k]=0$ for all $k$.

To prove this we begin with the following equation, verified by $M$: for any $X,Y\in\sln$ we have
\begin{equation}\label{eq:nij1}
M([X,Y])-M([M(X),Y]+[X,M(Y)])+[M(X),M(Y)]=0.
\end{equation}
This equation is an algebraic version of the vanishing Nijenhuis torsion equation, and follows by a straightforward computation from the fact that both $\mathfrak{so}_n$ and $\mathfrak b_n$ are subalgebras of \sln. Let now $1<k$ and consider the commutator $[T_1,T_k]$. We compute:
\[
\begin{aligned}
{}[T_1,T_k]&=[M(\tau_k([M(L),\dot L]),L]+[M(\tau_k(L)),[M(L),L]]\\
               &\quad-[M([M(\tau_k((L)),L]),L]-[M(L),[M(\tau_k(L)),L]],
\end{aligned}
\]
where we used the notation from the formula \eqref{eq:ddiff}. On the other hand, writing $L$ as a sum of two complementary projections induced by \eqref{eq:triangledec}, $L=M(L)+B(L)$ we have by the same formula \eqref{eq:ddiff} (since $[L,L]=0$)
\[
\tau_k([M(L),\dot L])=\tau_k([-B(L),\dot L])=[-B(L),\tau_k(L)]=[M(L),\tau_k(L)]-[L,\tau_k(L)].
\]
Hence
\[
\begin{aligned}
{}[T_1,T_k]&=[M([M(L),\tau_k(L)]),L]-[M([L,\tau_k(L)]),L]+[M(\tau_k(L)),[M(L),L]]\\
                 &\quad-[M([M(\tau_k((L)),L]),L]-[M(L),[M(\tau_k(L)),L]]\\
                 &=[[M(\tau_k(L)),M(L)],L]-[M([M(\tau_k((L)),L]),L]\\
                 &\quad-[M([\tau_k(L),M(L)]),L]+[M([\tau_k(L),L]),L]=0,
\end{aligned}
\]
where we used the skew symmetry and Jacobi identity for the commutator and the Nijenhuis equation \eqref{eq:nij1}.

One may try to show by this or a similar method that the fields $T_k(L)$ and $T_l(L)$ do in fact commute with each other for all $k$ and $l$, but so far our attempts to do this by hand did not bring result. Doing this might shed the light on the role of Nijenhuis type relations in the theory of Toda integrable systems: so far they appeared in the context of bihamiltonian structures and Lenart-Magri induction. In particular, it has been shown (see \cite{DM}) that the integrability of the full Kostant-Toda system can be derived from a bihamiltonian structure on it. On the other hand in the case of the full symmetric Toda system no bihamiltonian structures have been suggested so far.
\subsection{Other vector fields}
In the section \ref{sect:chopqr1} we described the formulas for the Hamiltonian fields, induced by the chopping integrals $I_{k,i}(L)$ (see equation \eqref{eq:nablaeki}). This formula resembles the formula for the fields $T_k(L)$, see \eqref{eq:matrixcoef1}; however even in the simplest case $k=2$, it is not easy to show that the Hamiltonian field, induced by $I_{k,n-1}(L)$ does coincide with $T_k(L)$ (we did it in section \ref{sect:hamilttk} by direct computation). Now, similarly to $T_k(L)$, induced by the matrix-valued functions $\tau_k(L)$, one may define another family of operator-valued functions $\theta_{k,m}(L)$:
\bea
\left(\theta_{k,m}(L)\right)^i_j=\sum_{s_1,\ldots,s_m}L^{s_1,\ldots,s_m,j,n-k+1,\ldots,n}_{1,\ldots,k,i,s_1,\ldots,s_m}/L_{1,\ldots,k}^{n-k+1,\ldots,n}
\eea
In effect, these functions are straightforward generalisations of the operation $\tau_k(L)$. Indeed, it follows from the formula \eqref{eq:matrixcoef1} that
\bea
M(\theta_{k,1}(L))=M(\tau_k(L)).\nn
\eea
Now by a straightforward, but lengthy and tedious calculation, which we omit, one can obtain the following proposition, similar to the Theorem \ref{theo:tauk}
\begin{Prop}
The hamiltonian flows given by hamiltonians $I_{k,n-m}(L)$ have Lax representation
\bea
\dot{L}=[M(\theta_{k,m}(L)),L],\nn
\eea
\end{Prop}
However, it is not easy showing by direct computation that $\theta_{k,m}(L)$ verify the same characteristic property \eqref{eq:charproperty}, that we need to prove the applicability of QR-decomposition method. It is even less clear, if one can prove the commutativity of the fields $\Theta_{k,m}(L)=[M(\theta_{k,m}(L)),L]$ without resorting to their being Hamiltonian fields of a commutative family.

\subsection{Vector fields due to non-chopping integrals}
It is well known that the full symmetric $sl(n)$ Toda system has additional integrals except for the integrals obtained by the chopping procedure ( \cite{ErcolFlaSin}, \cite{BG}, \cite{GS}, \cite{CS}), so that the Toda system is integrable in the non-commutative sense \cite{Neh}. Let us consider the example $n=4$ and
$$
X=
\left(
\begin{array}{c c c c}
 x_{11} & x_{12} & x_{13} & x_{14}\\
 x_{21} & x_{22} & x_{23} & x_{24}\\
 x_{31} & x_{32} & x_{33} & x_{34}\\
 x_{41} & x_{42} & x_{43} & x_{44}\\
\end{array}
\right).
$$
In this case the additional integral (via Chernyakov-Sorin method) has the following form
\bea
H=\frac {(X^2)^{34}_{12}}{X^{34}_{12}}.\nn
\eea
After taken the gradient, the symmetrization, and the projection on $\mathfrak{so}$ we get the matrix of M-operator in the following form:
\small{
$$
M_{J^{cs}}= \frac{1}{L^{34}_{12}}\left(
\begin{array}{rrrr}
0 & -B_{1} & 0 & 0\\
B_{1} & 0 & 0 & 0\\
0 & 0 & 0 &  -B_{2}\\
0 & 0 & B_{2} & 0\\
\end{array}
\right)
$$
and
\bea
J^{cs}=\frac {(L^2)^{34}_{12}}{L^{34}_{12}},\nn
\eea
where
$$\ B_{1}=(L^2)_{14}a_{13}-(L^2)_{13}a_{14}, \ B_{2}=(L^2)_{14}a_{24}-(L^2)_{24}a_{14}$$
and $L$ as usually is
$$
L=
\left(
\begin{array}{c c c c}
 a_{11} & a_{12} & a_{13} & a_{14}\\
 a_{12} & a_{22} & a_{23} & a_{24}\\
 a_{13} & a_{23} & a_{33} & a_{34}\\
 a_{14} & a_{24} & a_{34} & a_{44}\\
\end{array}
\right).
$$
Considering skew-symmetric square of L which has $2\times2$ minors of L as entries, after chopping procedure the characteristic polynomial $\widetilde{\chi}_{1}(L,\lambda)$
will have the following form:
\small{
$$
\left |
\begin{array}{ccccc}
 a_{11} a_{23}-a_{12} a_{13} & a_{11} a_{33}-a_{13}^2  -\lambda & a_{11} a_{34}-a_{13} a_{14} & a_{12} a_{33}-a_{13} a_{23} & a_{12} a_{34}-a_{14}
   a_{23}\\
 a_{11} a_{24}-a_{12} a_{14} & a_{11} a_{34}-a_{13} a_{14} & a_{11} a_{44}-a_{14}^2  -\lambda & a_{12} a_{34}-a_{13} a_{24} & a_{12} a_{44}-a_{14}
   a_{24}\\
 a_{12} a_{23}-a_{13} a_{22} & a_{12} a_{33}-a_{13} a_{23} & a_{12} a_{34}-a_{13} a_{24} & a_{22} a_{33}-a_{23}^2  -\lambda & a_{22} a_{34}-a_{23}
   a_{24}\\
 a_{12} a_{24}-a_{14} a_{22} & a_{12} a_{34}-a_{14} a_{23} & a_{12} a_{44}-a_{14} a_{24} & a_{22} a_{34}-a_{23} a_{24} & a_{22}
   a_{44}-a_{24}^2  -\lambda \\
 a_{13} a_{24}-a_{14} a_{23} & a_{13} a_{34}-a_{14} a_{33} & a_{13} a_{44}-a_{14} a_{34} & a_{23} a_{34}-a_{24} a_{33} & a_{23}
   a_{44}-a_{24} a_{34}
\end{array}
\right |,
$$
}
and we get the additional integral $J_{1,2} = \frac{\widetilde{E}_{1,2}(L)}{\widetilde{E}_{1,0}(L)}$. This integral and the integral $J^{cs}$ are functionally dependent.
\beq{spectral curve}
\begin{array}{c}
\widetilde{E}_{1,2} =
\left(a_{14} a_{23}-a_{13} a_{24}\right) \left( a_{11} a_{33}+a_{22}a_{33}+a_{11}a_{44}+a_{22} a_{44}-a_{13}^2-a_{23}^2-a_{14}^2-a_{24}^2 \right)+\\
+\left(a_{11} a_{23}-a_{12} a_{13}\right) \left(a_{13} a_{34}-a_{14} a_{33}\right)
+\left(a_{12} a_{23}-a_{13}a_{22}\right) \left(a_{23} a_{34}-a_{24} a_{33}\right)+\\
+\left(a_{11} a_{24}-a_{12} a_{14}\right) \left(a_{13} a_{44}-a_{14} a_{34}\right)
+\left(a_{12} a_{24}-a_{14} a_{22}\right) \left(a_{23} a_{44}-a_{24} a_{34}\right),\\
\ \\
\widetilde{E}_{1,0} =a_{13} a_{24}-a_{14} a_{23}.
\end{array}
\eq
$$J_{1,2}=J^{cs}-\frac{1}{2}(Tr(L^{2})-(Tr(L))^{2})=J^{cs}+E_{2},$$
where $E_{2}$ is the coefficient of this characteristic polynomial $\chi_{0}(L,\lambda)$, see Appendix A. The connection between vector fields $\mathcal{T}_{H}$ on symmetric matrix space generated by integrals $J_{1,2}, \ J^{cs}, \ (1/2)TrL^{2}$ has the following form:
\beq{Flows}
\mathcal{T}_{J_{21}}(L)=\mathcal{T}_{(1/2)TrL^{2}}(L)+\mathcal{T}_{J^{cs}}(L),
\eq
and the dynamics associated with these integrals has the following Lax representation
$$\frac{dL}{dt_{J_{1,2}}}=[M(L),L]+[M_{J^{cs}},L], \ M(L)=L_{+}-L_{-}.$$

\subsection{Further questions}
Let us finally formulate few questions, which are closely related with the material covered in this paper and which may become the subject of future investigations.

First of all, one can extend the constructions of Toda system and its integrals to other representations of the Lie group $SL_n(\R)$. For instance, take the exterior square representation of $SL_4(\R)$ in $SL_6(\R)$; then one can use it to pull back the integrals of the Toda system on $\mathfrak{sl}_6$ to $\mathfrak{sl}_4$. Direct computations show that these pull-backs will contain some of the integrals of the Toda system on $\mathfrak{sl}_4$. 
Another example of the way representations of the Lie algebra show up in the study of the Toda system, can be found in \cite{CSS19} (c.f. the section \ref{sect:CSRepr}). One can ask about the way different representations can be used to obtain the integrals of the Toda system, in particular, about the relation between the representation and the number of integrals that can be ``pulled'' from it; the way the integrals change when some natural operations are applied to representations is also an open question.

Still more generally one can do the following trick: let $\rho:SL_n(\R)\to SL_N(\R)$ be a representation, which when restricted to the Lie algebra preserves the symmetricity of the matrices, see previous example; restricting it to $Symm_n$ we can use it to induce a dynamical system on $Symm_n$ from the Toda system on $Symm_N$. The properties of these systems, their integrability and relation with Toda system on $Symm_n$ seem to be an interesting subject for future work. This question is closely related with the problem of describing all possible Lax representations of the full symmetric system: is it possible that in addition to the equation \eqref{eq:laxtoda}, one can write similar equations
\[
\frac{d\tilde L}{dt}=[\tilde M,\tilde L]
\]
for some other matrices $\tilde L$ and $\tilde M$, depending on the point in our phase space $Symm_n$, such that the resulting system will coincide with the Toda system. In this case one can obtain the first integrals of the system as the invariant functions of $\tilde L$ matrix. It is interesting, if the chopping integrals have this nature.

Another question, that one can ask, is whether the constructions used in this paper and the corresponding results stay valid when we pass from the classical to quantum Toda system: one can suggest at least two such noncommutative (quantum) generalisations. One consists with replacing the functions on $\sln^*$ with the universal enveloping algebra of \sln. In that case certain progress has been made (see for instance \cite{T1}): one can obtain analogs of the chopping integrals there as the elements in a suitable localisation of the universal enveloping algebra $U\sln$. Still, the question, if these elements can be obtained by a procedure, similar to the one we used in lemma \ref{lem:key} remains widely open. The other question is still more intriguing: it is known, that the usual Toda system (open Toda chain), can be generalised to the situation, where the ``dynamics'' takes place at an arbitrary ring with differential $d$; in this case one can find the ``solutions'' of such system in terms of quasi-determinants and use them to describe solutions of the noncommutative Painlev\'e equation, see  \cite{BRRS} for example. In our case, one can consider the ``system'' on matrices with noncommutative entries of arbitrary nature; the question, if this construction preserves any resemblance with the one we studied here, is open.

\appendix
\section{Example of the "chopping" construction using $D_{k}$, case $n=4$}
Here we demonstrate the construction in the case $n=4$. First of all let us consider characteristic polynomial without chopping procedure. It has the following form:
\beq{spectral curve}
\begin{array}{c}
\chi_{0}(L,\lambda)= \lambda^{4}E_{0} - \lambda^{3}E_{1} + \lambda^{2}E_{2}
 - \lambda E_{3} + E_{4},\\
\ \\
E_{0}=1, \ E_{1}=TrL, \ E_{2}=\frac{1}{2}(TrL)^{2} - \frac{1}{2}TrL^{2},\\
\ \\
E_{3}=\frac{1}{3}(TrL)^{3} - \frac{1}{3}TrL^{3}-TrL(\frac{1}{2}(TrL)^{2} - \frac{1}{2}TrL^{2}), \ E_{4}=detL,
\end{array}
\eq
where $TrL$ is Casimir and
$$
L=
\left(
\begin{array}{c c c c}
 a_{11} & a_{12} & a_{13} & a_{14}\\
 a_{12} & a_{22} & a_{23} & a_{24}\\
 a_{13} & a_{23} & a_{33} & a_{34}\\
 a_{14} & a_{24} & a_{34} & a_{44}\\
\end{array}
\right).
$$
Now we consider the characteristic polynomial $\chi_{1}(L,\lambda)$
\beq{spectral curve2}
\begin{array}{c}
\chi_{1}(L,\lambda)=\lambda^{2}E_{1,0} + \lambda E_{1,1} + E_{1,2},\\
\ \\
E_{1,0} = a_{14}, \,\,\,\, E_{1,1} = A_{\frac{34}{13}} + A_{\frac{24}{12}}, \,\,\,\, E_{1,2} = A_{\frac{234}{123}},,
\end{array}
\eq
Then we have to put $k=1$ and the formula (\ref{eq:altchop1}) will give:
$$D_1(E_{0})=0, \ D_1(E_{1})=0, \ E_{1,0}=D_1(E_{2}), \ E_{1,1}=D_1(E_{3}), \ E_{1,2}=D_1(E_{4}),$$

$$E_{k,m}=D_k(E_{2k+m}),$$
where we put $D_1=\frac{\partial}{\partial x^4_1}$ before the symmetrization.

Note, that each one of the DNLT-integrals which constructs by $E_{k,m}$ can be expressed by the integrals obtained by Chernyakov-Sorin method (the reverse is not true) so these integrals are functionally dependant, for example  \beq{twointegrals}
I^{cs}_{1,2}=I_{1,2}-E_{2}+E_{1}(E_{1}+I_{1,1}),
\eq
where
$$I_{1,1}(L) = \frac{E_{1,1}(L)}{E_{1,0}(L)},$$
$E_{2}=\frac{1}{2}(TrL)^{2} - \frac{1}{2}TrL^{2}$. Note that $I^{cs}_{1,1}$ and $E_{1}=TrL$ are Casimir operators and $I^{cs}_{1,1}=E_{1}+I_{1,1}$.

\section{Calculation the matrix of M-operator, case $n=4$}
Considering AKS method for the full symmetric Toda system the dynamic follows from the Lie-Poisson brackets on symmetric matrices in the following form:
\beq{LaxandM}
\begin{array}{c}
\frac{dL}{dt_{H}}=[M(\nabla H), L], \ \nabla H \in \mathfrak{sl}_n,
\end{array}
\eq
where M-operator is projection on $\mathfrak{so}$. Here we assume the symmetrization after taking the gradient. Let us consider
$$H=I^{cs}_{1,2}- (1/2)TrX^{2}=\frac {(X^3)_{61}}{x_{61}}-(1/2)TrX^{2}, \ X \in \mathfrak{sl}^{\ast}_n$$
and take the gradient:
$${F^{ab}}=\frac{\partial H}{\partial x_{ab}}=\frac{\partial}{\partial x_{ab}} \left( \frac{x_{4i} x_{ij} x_{ji} - (1/2)x_{kl} x_{lk} x_{41}}{x_{41}} \right).$$
We get
$$
{F^a_b}=\frac 1 {x_{41}} \left(
\begin{array}{rrrr}
\ast & 0 & 0 & 0\\
\ast & \ast & X^{34}_{12} & 0\\
\ast & \ast & \ast & 0\\
\ast & \ast & \ast & \ast\\
\end{array}
\right),
$$
where $\ast$ are some expression of $x_{ij}$ which canceled after projection on $\mathfrak{so}$ along $\mathfrak{b_{-}}$. After the symmetrization and the projection on $\mathfrak{so}$ we get the matrix of M-operator in the following form:
$$
M_{H}=\frac 1 {a_{14}} \left(
\begin{array}{rrrr}
0 & 0 & 0 & 0\\
0 & 0 & L^{34}_{12} & 0\\
0 & -L^{34}_{12} & 0 & 0\\
0 & 0 & 0 & 0\\
\end{array}
\right)
$$
which coincides with the matrix of M-operator of DLNT integral
$$I_{1,2} = \frac{E_{1,2}(L)}{E_{1,0}(L)}$$
which can be obtained by the direct calculation too.
And finally we have the formula:
$$\frac{dL}{dt_{I^{cs}_{1,2}}}=[M(L),L]+[M_{I_{1,2}},L],$$
where $M(L)=L_{+}-L_{-}=M_{(1/2)TrL^{2}}$.

\end{document}